\documentclass[11pt]{article}
\usepackage[legalpaper,  margin=1.5in]{geometry}
\usepackage[yyyymmdd,hhmmss]{datetime}
\usepackage{amsmath,amssymb,amsfonts,amsthm}
\usepackage{authblk}
\usepackage{stackengine}
\usepackage{verbatim}
\usepackage{enumerate}
\usepackage{tikz}
\usepackage{algorithm}
\usepackage{algpseudocode}

\usepackage{url}
\PassOptionsToPackage{hyphens}{url}
\PassOptionsToPackage{hyphens}{hyperref}
\makeatletter
\g@addto@macro{\UrlBreaks}{\UrlOrds}
\makeatother

\usepackage{graphicx}

\allowdisplaybreaks

\usepackage[toc]{glossaries}

\usepackage{tocbibind}

\usepackage{lineno}

\usepackage{algorithm,algorithmicx,algpseudocode}



\floatname{algorithm}{Procedure}

\renewcommand{\algorithmicreturn}[1]{\bgroup\\  ~#1\egroup}
\renewcommand{\algorithmiccomment}[1]{\bgroup\hfill//~#1\egroup}

\newcommand{\thicktilde}[1]{\mathbf{\tilde{\text{$#1$}}}}

\theoremstyle{plain} \newtheorem{lemma}{\textbf{Lemma}}
\theoremstyle{plain} 
\theoremstyle{plain} \newtheorem{remark}{\textbf{Remark}}
\theoremstyle{plain} \newtheorem{theorem}{\textbf{Theorem}}
\theoremstyle{plain} 
\theoremstyle{plain} 
\theoremstyle{plain} \newtheorem{corollary}{\textbf{Corollary}}
\theoremstyle{plain} 
\theoremstyle{plain} \newtheorem{definition}{\textbf{Definition}}
\theoremstyle{definition}

\makeatletter
\newcommand{\pushright}[1]{\ifmeasuring@#1\else\omit\hfill$\displaystyle#1$\fi\ignorespaces}
\newcommand{\pushleft}[1]{\ifmeasuring@#1\else\omit$\displaystyle#1$\hfill\fi\ignorespaces}
\makeatother

\makeatletter
\let\@@pmod\pmod
\DeclareRobustCommand{\pmod}{\@ifstar\@pmods\@@pmod}
\def\@pmods#1{\mkern4mu({\operator@font mod}\mkern 6mu#1)}
\makeatother

\newcounter{NbCogito} 

\newcounter{NbFactum} 

\newcounter{NbTabulare} 

\newcounter{NbCoco} 

\newcommand{\rmv}[1]{}

\def\cA{{\mathcal A}}

\def\cB{{\mathcal B}}
\def\cC{{\mathcal C}}
\def\cD{{\mathcal D}}
\def\cE{{\mathcal E}}

\def\cP{{\mathcal P}}

\def\cT{{\mathcal T}}

\def\FF{{\mathbb F}}

\def\ZZ{{\mathbb Z}}

\newenvironment{myproof}[1][\myproofname]{\par
\normalfont \topsep6pt\relax
\trivlist
\item[\hskip\labelsep
\itshape
#1.]\ignorespaces
}{%
\endtrivlist\hfill$\square$
}
\providecommand{\myproofname}{Proof}

\title{Sharing the Path: A Threshold Scheme from Isogenies and Error Correcting Codes}
\author{Mohamadou Sall}
\affil
{\stackunder{\small{Department of Electrical and Computer Engineering, University of Waterloo, Canada}}
{\mbox{\small{\texttt{msall@uwaterloo.ca}}}}
}

\author{M. Anwar Hasan}
\affil
{\stackunder{\small{Department of Electrical and Computer Engineering, University of Waterloo, Canada}}
{\mbox{\small{\texttt{ahasan@uwaterloo.ca}}}}
}

\providecommand{\keywords}[1]
{
  \textbf{Keywords:} #1
}

\date{}

\begin{document}

\maketitle

\begin{abstract}
In 2022, a prominent supersingular isogeny-based cryptographic scheme, namely SIDH, was compromised by a key recovery attack. However, this attack does not undermine the isogeny path problem, which remains central to the security of isogeny-based cryptography. Following the attacks by Castryck and Decru, as well as Maino and Martindale, Robert gave a mature and polynomial-time algorithm that transforms the SIDH key recovery attack into a valuable cryptographic tool. In this paper, we combine this tool with advanced encoding techniques to construct a novel threshold scheme.
\end{abstract}

\keywords{Secret Sharing, Threshold Scheme, Isogeny Path Problem.}

\section{Introduction}

Blockchain technology has sparked widespread attention since its anonymous introduction in 2009 \cite{N08}. It offers a novel paradigm in which transactions are maintained in a decentralized data structure and as such it is widely used in cryptocurrencies, including Bitcoin and Ethereum. In blockchain applications secret sharing schemes (SSS) are used extensively, in particular in electronic voting, data storage, and supply 
chain management. 

\begin{definition}\label{def:SSS}
A secret sharing scheme consists of a dealer, a group $\cP=\{\cP_1, \cdots, \cP_n \}$ of $n$ participants, a secret space $S$ with a uniform probability distribution, $n$ share spaces $S_1,\cdots, S_n$, a share and a secret recovering procedures $\mathcal{S_P}$ and $\mathcal{R_P}$.
\end{definition}
On the other hand, threshold schemes allow the distribution of a secret key into shares among multiple parties or devices, such that only a set of authorized parties can jointly recover the secret and perform cryptographic operations such as signing and decrypting data.  
\begin{definition}
A secret sharing scheme is called a $(n, t)$-threshold scheme, where $1\leq t\leq n$ if the following conditions are satisfied: (1) the number of participants is $n$; (2) any subset of $t$ participants can obtain all information about the secret with their shares; and (3) any subset of $t-1$ or fewer participants cannot get any information about the secret with their shares.
\end{definition}
Since their independent introduction by Shamir \cite{S79} and Blakley \cite{B79}, secret sharing and threshold schemes have been widely used in many classical 
and post-quantum cryptography
protocols. 
The main advantage of their design is that key recovery attacks on threshold schemes require more effort than on the non-threshold ones, as the adversary has to attack more than one device or party simultaneously. The ability to distribute trust among several parties has sparked interest from governmental bodies such as NIST \cite{NISTthresh} and driven the generalization of Shamir's idea, which relies on the Lagrange interpolation, to many other fields. Indeed, as early as 1981, MacElliece and Sarwate \cite{MS81} established a connection between Shamir secret sharing and error-correcting codes.\\

\noindent \textbf{Error Correcting Codes and Secret Sharing Schemes:}
An $[\ell, k, d]$ linear code $\cC$ of length $\ell$ over the finite field $\FF_{q^m}$ is a $k-$dimensional subspace of $\FF_q^m$ with minimum (Hamming) distance $d$. 
The dual code of $\cC$ is defined to be its orthogonal subspace $\cC^{\perp}$ with respect to the Euclidean inner product.
\begin{theorem}[\cite{MS78}]
    A code $\cC$ with minimum distance $d$ can correct $\frac{1}{2}(d-1)$ errors. If $d$ is even, the code can simultaneously correct $\frac{1}{2}(d-1)$ errors and detect $d/2$ errors.
\end{theorem}

An erasure is an error for which the error location is known but the error magnitude is not known. It is easier to correct an erasure because we don't need the detection process. Indeed, a code $\cC$ with minimum distance $d$ can correct up to $(d-1)$ erasures. The correction capacity of a linear code $\cC$ depends on its minimum distance but it is also desirable to have a large dimension $k$, which means the code can transmit large amounts of information. However, for a fixed length $\ell$, it is difficult to have large $d$ and $k$ due to the {\it singleton bound}, which is 
$$d\leq \ell-k+1.$$
Codes with $d = \ell - k + 1$ are called {\it maximum distance separable (MDS) or optimal} codes. This class of codes or codes with large minimum distance plays a central role in cryptography, particularly in secret sharing and threshold schemes. For example, in \cite{M93, M95} Massey gave a code-based secret sharing scheme where the threshold $t$ is related to the minimum distance $d^{\perp}$ of the dual code $\cC^{\perp}$. Indeed as shown by Renvall and Ding in \cite{RD96} in Massey's construction any set of $d^{\perp}-2$ or fewer shares does not give any information on the secret, and there is at least one set of $d^{\perp}$ shares that determines the secret. Van Dijk pointed out in \cite{D95} that Massey's approach is a special case of the construction introduced by Bertilsson and Ingemarson in \cite{BI92} that uses linear block codes. We have various code-based secret sharing schemes in the literature with notable similarities and generalizations. Some constructions modify the type of codes \cite{WD14}, others improve on crucial operations e.g., XOR \cite{CLM16} 
to gain efficiency, and some leverage on the encoding system \cite{HB22} 
for more performance. The aim of this paper is to combine the properties of error-correcting codes and supersingular elliptic curves torsion points to share a secret isogeny path.\\

\noindent \textbf{Elliptic Curves and Secret Sharing Schemes:} Elliptic curve cryptography (ECC) relies on the hardness of the discrete logarithm problem known as  ECDLP. Due to its versatile application (e.g., in secure DNS, TLS, Bitcoin),
there have been significant efforts in proposing DLP-based threshold schemes. 
However, these cryptographic algorithms can not withstand attacks from quantum computers. The NIST post-quantum cryptography competitions 
aims to find replacements for the current number-theoretic solutions based on integer factoring and discrete logarithms. 
Hard problems based on fields such as coding theory, lattices, and multivariate polynomials are believed to be good candidates. In \cite{CS19}, Cozzo and Smart observed that multivariate-based schemes LUOV \cite{BPSV19} and Rainbow \cite{DCPSY19} allow for a natural threshold construction and recently del Pino et al. \cite{PKMMPS24} constructed Raccoon, a threshold signatures scheme from lattices.\\
\indent Another solution to the threat of quantum computers is maps between elliptic curves i.e., isogenies. Despite the attacks \cite{CD23, MM22, R23} made on SIDH and SIKE, supersingular isogeny-based cryptography remains fully trustworthy. 
Indeed, these are key recovery attacks that do not affect the hard problem underlying isogeny-based cryptography. That is the problem of finding a path between two given elliptic curves or equivalently computing the endomorphism ring of an elliptic curve. 
For example, the work in \cite{CD23, MM22, R23} does not undermine the security of CSIDH and SQIsign, which is a prominent secure post-quantum isogeny-based scheme that recently advanced to the second round of NIST additional post-quantum digital signature scheme competition.\\
\indent Due to their small key size and practical applications in constrained devices, isogeny-based cryptography has gained a lot of attention and recently many schemes have been designed to give their threshold versions. In \cite{FM20}, Feo and Meyer adapted Desmedt and Frankel’s scheme \cite{D93} to the Hard Homogeneous Spaces (HHS) framework of Couveignes \cite{C06}. They initiate the study of threshold schemes based on isogenies and as a follow up we encounter many interesting threshold schemes, including Sashimi \cite{CS20}, SCI-RAShi \cite{BDPV21}, and SCI-Shark \cite{ABCP23}. These schemes are initially related to the action of a group $G$ on a set of elliptic curves $\cE$. Indeed, in Feo and Meyer design, the goal of the participants $\cP_1, \cP_2, \cdots, \cP_t$ is to use their secret shares $s_i$  to evaluate the group action $[s]E_0$, where $s\in \ZZ/q\ZZ$ is a secret and $E_0$ a given elliptic curve in $\cE$. \\

\noindent \textbf{Our Main Goal:} In this paper, we focus on a somewhat different strategy: thresholdizing isogeny-based cryptography using torsion points. Since the secret isogeny path is the backbone of post-quantum algorithms using supersingular elliptic curves, to ensure more security, it would be interesting to share this secret among $n$ participants with a threshold of recovery $t\leq n$. As we mentioned, there are many isogeny-based threshold schemes but to the best of our knowledge, this is the first technique that uses the attack on SIDH as a tool and shares the isogeny path upon which many post-quantum encryptions and signatures schemes rely.\\

\noindent {\bf Organization:} In Section \ref{sec:prelim}, we give preliminaries on isogeny representation and encoding functions over elliptic curves. We particularly highlight how the attack on SIDH/SIKE becomes a valuable cryptographic tool. In Section \ref{sec:TuTPE}, we give our main results and the algorithms to share and recover a secret using isogenies and encoding techniques. In Section \ref{sec:SS}, we analyze the security of those algorithms. In Section \ref{sec:SCS}, we extend our study to particular types of error-correcting codes that fit better our study. This extension leads to corollaries derived from the main theorem using Reed-Solomon codes, in Subsection \ref{subsec:BRSIS}, and burst erasure corrections, in Subsection \ref{subsec:BESIR}.

\section{Preliminaries}\label{sec:prelim}


\subsection{Isogeny Representation: From an Attack to a Tool}\label{subsec:IRAT}

Let $E_0$ and $E_1$ be elliptic curves over a finite field $\FF_q$. An isogeny from $E_0$ to $E_1$ is a surjective morphism $$I: E_0 \longrightarrow E_1$$ satisfying $I(O)=O$ where $O$ is the identity element of the elliptic curve's group. Two elliptic curves $E_0$ and $E_1$ are isogenous if there is an isogeny from $E_0$ to $E_1$. The hard problem behind isogeny-based cryptography is stated as follows.\\

\noindent
{\bf Isogeny Path Problem (IPP).}
Given two isogenous elliptic curves $E_0$ and $E_1$ over a finite field $\FF_q$, compute an isogeny from $E_0$ to $E_1$.\\

Many cryptographic schemes based on isogenies (some still secure others not) between ordinary  
and supersingular elliptic curves have been studied in the literature. These schemes are attractive because of their compacity and their small key sizes. This has been one of the most valuable advantages of SIKE, which made it to the 4th round of the NIST competition.
Unfortunately, the algorithm was subject to a \textit{key recovery attack} \cite{CD23, MM22, R23} and the SIKE team acknowledged that SIKE and SIDH from which it is derived are insecure and should not be used. However, it is important to highlight that despite this attack isogeny-based cryptography remains fully trustworthy. We explain this trustworthiness in the following.\\
\indent The high-level idea of supersingular Diffie-Hellman key exchange protocol is described by Jao and Feo \cite{JF11} as follows. To communicate securely Alice and Bob select respectively two secret isogenies $I_A: E_0 \rightarrow E_A $ and $I_B: E_0 \rightarrow E_B$.
\begin{itemize}
    \item They fix the respective degrees $N_A$ and $N_B$ of the isogenies $I_A$ and $I_B$, with $N_A$ and $N_B$ two coprime integers.
    \item They choose two bases $\cB_A$ and $\cB_B$ respectively for $E_0[N_A]$, the $N_A-$torsion of $E_0$, and $E_0[N_B]$, the $N_B-$torsion of $E_0$.
    \item Alice reveals the image of $\cB_B$ and Bob reveals the image of $\cB_A$.
\end{itemize}
As we can see, SIDH (and SIKE) does not directly rely upon the IPP problem but it uses the following `mild' version of this problem.\\ 

\noindent
{\bf Supersingular Isogeny with Torsion Problem (SI-TP).}
Let $E_0$ be a supersingular curve defined over $\FF_{q}$ with $q=p$ or $q=p^2$. Set $E_0[N] \ = \ \langle P, Q \rangle$. Let $I: E_0 \rightarrow E_1$ be an isogeny of degree $d$ and let $P' = I(P)$ and  $Q' = I(Q)$. Given $E_0, \ P, \ Q, \ E_1, \ P'$ and $Q'$, compute the isogeny $I$.\\

The work in \cite{CD23, MM22, R23} breaks an instance of the SI-TP problem but, not the more general (supersingular) isogeny path problem IPP. It does not apply to protocols like CSIDH \cite{CLMPR18} and SQIsign and its variants, whose security reduces to a distinguishing problem on isogeny walks generated according to an ad-hoc distribution. Furthermore,  as shown by Feo, Fouotsa, and Panny \cite{FFP24}, we can also have other variants of the IPP problem that can be hard depending on the chosen parameters. The key hardness problem of many isogeny-based protocols is based on the difficulty of recovering a large degree isogeny $I: E_0 \rightarrow E_1$ between two ordinary or supersingular elliptic curves. Without more information on $E_0$ and $E_1$, like images of torsion points or an explicit representation of part of their endomorphism rings, this problem still has exponential quantum security for supersingular curves.\\
\indent The successful attack on SIDH/SIKE gave a U-turn in the isogeny-based cryptography field. It made an emphasis on well-known secure schemes such as SQIsign and gave rise to a valuable cryptographic tool. When $E_0$ is a random curve, the dimension 2 attack of Maino and Martindale \cite{MM22} and Castryck and Decru \cite{CD23} are in heuristic subexponential time. To gain efficiency
Robert \cite{R23} extended these attacks using Kani's lemma and Zarhin's trick. He gave, for the first time \cite{R22} a precise complexity estimate, and a polynomial time attack (with or without 
a random starting curve) with no heuristics. We summarize Robert's efficient attack on SIDH in the following theorem.
\begin{theorem}\label{theo:IsoEfficiency}
    Let $N$ be a $polylog(q)$-smooth integer and $I: E_0 \rightarrow E_1$
    be an isogeny of degree $d < N^2$. Let $\langle P, Q\rangle$ be a basis of $E[N]$ such that $P$ and $Q$ can be represented over a small extension of $\FF_q$. Assume that we are given the image of $I$ on the basis $\langle P, Q\rangle$. Then an instance of the SI-TP problem can be solved in time $polylog(q)$.
\end{theorem}
In more details, for $\ell_N$ a large prime factor of $N$, Robert \cite[Section 2]{R23} proved that the dimension 8 attack costs $\thicktilde{O}(\log N\ell_N^8)$ arithmetic operations, where  $\thicktilde{O}(x)$ is a shorthand for $O(x polylog(x))$. Furthermore, in some refine cases such as when $\ell_N=O(1)$, or $\ell_N = O(\log \log N)$, the attack is in quasi-linear time, i.e., in $\thicktilde{O}(\log N)$ arithmetic operations in $\FF_q$. So it became as efficient asymptotically as the key exchange itself. This is how the attack turns into a valuable cryptographic tool that is currently used in various schemes. 
In the context of our study, we aim to use this tool alongside encoding techniques to securely share the secret path.

\subsection{Encoding on Elliptic Curves}\label{subsec:EEC}

Many cryptographic protocols using elliptic curves are related to the following question:
how can one represent a uniform point on an elliptic curve $E(\FF_q)$ in a public 
or a private way as a close to uniform random bit string?\\
\indent The standard technique of encoding a point $(x, y)\in E(\FF_q)$ is to compress it in order to include only the $x-$coordinate and a single bit indicating the parity of $y$. It is very simple but the problem is that one cannot reliably obtain strings in this representation
since $x(E(\FF_q))$ covers only about half of the elements of $\FF_q$. Furthermore, this technique is vulnerable to censorship since it generates strings distinguishable from random bits. As a countermeasure, Bernstein et al. \cite{BHKL13} used an injective function $$\psi^{-1} : \FF_q \rightarrow E(\FF_q)$$ that maps a subset $S \subset \FF_q$ of length $\approx q/2$ into $E(\FF_q)$. So, if $P\in \psi^{-1}(S)$ then $P$ can be represented by its unique preimage in $S$. Mainly, they constructed Elligator 1 and Elligator 2 functions such that if $q$ is close to a power of $2$ one can have uniform random bit strings. In this paper, we refer only to Elligator 2 since it is more widely applicable than Elligator 1. The encoding function Elligator 2 has the advantage of being very efficient, but suffers from the following limitations: 
\begin{itemize}
    \item Supported curves must have even cardinalities.
    \item Even if we get the right curve, the probability of a random point having a bit-string preimage is about $1/2$. The rejection sampling process is not expensive but if a protocol requires a particular curve point to be communicated as a $b$-bit string there is no reason to think that this point can be expressed in string form.
\end{itemize}
To address these shortcomings Tibouchi presented Elligator Squared \cite{T14}. The basic idea of his construction is to randomly sample a preimage of the point $P$ to be encoded under an admissible encoding of the form $$(x, y) \mapsto f(x) + f(y),$$ where $f$ is an algebraic encoding. The main advantage is that such function $f$ is known to exist for all elliptic curves and induces a close to uniform distribution on the curve. However, we should note that the output size is about twice as large as Elligator and the representation function is more costly compared to Elligator. In 2022 at Asiacrypt, ElligatorSwift, a faster version of Elligator Squared was presented by  Chavez-Saab et al. \cite{CRT22}.\\
\indent  Each of the enumerated encoding algorithms in this section comes with its own pros and cons. For example, the simple standard technique may be useful if the required protocol does not need uniformity of bit-strings and there are fewer security issues. The Elligator 2 function gives rise to efficient bijection for essentially all curves $E$ for which $\#E(\FF_q)\equiv 0 \mod2$. Elligator Squared and its improvement \cite{TK17, CRT22} are suitable for general-purpose elliptic curves.

\begin{remark}
 In the context of our study, we denote by $\psi$ the encoding function, and in reference to the different algorithms cited above we name the corresponding algorithm as \textit{SESS} (for Standard Elligator Swift Squared) algorithm. For an elliptic curve $E$ over a finite field $\FF_q$ and a point $P\in E,$ \textit{SESS($E, P$)} outputs a chain of bits, the encoding of $P$, using the appropriate and the fastest algorithm between the standard encoding, the Elligator 2 function, and the Elligator Squared and its improvement.
\end{remark}

\section{Thresholdization using Torsion Points and Encoding}\label{sec:TuTPE}

In this section we give our main results, the algorithms of sharing and recovering a secret following Definition \ref{def:SSS}. These algorithms are mainly based on the SI-TP problem and Theorem \ref{theo:IsoEfficiency} that efficiently gives its solution. To gain more flexibility and have fewer parameters (at the expense of some restrictions) we rather use another version of this problem on an ad hoc basis. This version is derived from the following question:
\begin{center}
``Is it possible to solve, in polynomial time, the SI-TP problem with one point?"
\end{center}
that had been of interest within the scientific community. In \cite{FFP24}, Feo et al gave an affirmative answer to this question, which we reformulate in the following.\\

\noindent
{\bf Supersingular Isogeny  with One Torsion Problem (SI-OTP)}
    Let $I: E_0\rightarrow E_1$ be a $d-$isogeny between two supersingular elliptic curves $E_0$ and $E_1$ over $\FF_q$ with $q=p$ or $q=p^2$. Let $P\in E(\FF_q)$ be a point of order $N$ and $P'=I(P)$. Given $E_0, \ P, \ E_1,$ and $P'$ compute $I$.

\begin{theorem}[\cite{FFP24}]\label{theo:SI-OTP}
    Let $q$ be a power of a prime $p$. Let $N$ and $d$ be two integers such that $N$ is coprime to $d$. Let $I: E_0\rightarrow E_1$ be a degree $d$ isogeny between two supersingular elliptic curves $E_0$ and $E_1$ over $\FF_q$. If $N$ contains a large smooth square factor then we can resolve the SI-OTP problem in polynomial time.
\end{theorem}

The basic idea is to reduce the SI-OTP problem to the well-known SI-TP problem. If 
$N$ contains a large smooth square factor, for example, if $N$ is a power of a small prime $\ell$, then this reduction can be done in  $poly(\ell)$-time. After the reduction, we use Robert’s technique to recover the isogeny.\\
\indent In the following $$\psi: E(\FF_q)\longrightarrow \{0, 1\}^{\gamma n/2}$$ is the encoding function, associated with \textit{SESS}, which maps a point $P$ of $E$ to a fixed string of bits of length $\gamma n/2$. $\cP$ is the set of participants and  
$$f:\{0, 1\}^{\gamma n}\longrightarrow \cP$$ 
a sharing function that distributes $ \gamma n$ bits to $n$ participants. Each participants $\cP_i$ gets $\{i, e_i\}$
as a secret share where $i$ is an index and $e_i$ a chain of bits of length $\gamma$. In other words, for a sequence of bits $$(b_0, \ b_1, \ b_2, \ b_3, \ b_4, \ b_5, \ \cdots, \ b_{\gamma n-1}),$$ the function $f$ distributes the bits among the $n$ participants according to the partition specified in Table \ref{tab:sharing_bits}.

\begin{table}
    \centering
    \begin{tabular}{|c|c|}
    \hline
        Block of bits of length $\gamma$ & Participants \\
    \hline    
    $b_0\cdots b_{\gamma-1}$     & $\cP_0$ \\
    \hline
    $b_{\gamma}\cdots b_{2\gamma-1}$   &  $\cP_1$ \\
    \hline
        $\hdots$ & $\hdots$ \\
    \hline
    $b_{\gamma(n-1)}\cdots b_{\gamma n-1}$   &  $\cP_{n-1}$ \\
    \hline
    \end{tabular}
    \caption{Sharing bits to the $n$ participants.}
    \label{tab:sharing_bits}
\end{table}

\begin{theorem}\label{theo:main}
    Let $I: E_0\longrightarrow E_1$ be a secret $d-$isogeny between two supersingular elliptic curves $E_0$ and $E_1$ over $\FF_q$. Let $N$ be an integer containing a large smooth square. Let $P\in E(\FF_q)$ be a point of order $N$ and $P'=I(P)$. Let $\cC$ be a code over $\FF_2$ that can correct up to $\gamma(n-t)$ erasures.  Then one can share the path $I$ between $n$ participants with a threshold $t$.
\end{theorem}

\begin{proof}
    We give the secret isogeny $I$ and the extra pieces of information $P\in E_0$ and $P'\in E_1$ to the dealer. Let $$S_P = SESS(E_0, P) \text{ and }
    S_{P'} = SESS(E_1, P').$$ Then $w=S_P||S_{P'}$, the concatenation of $S_P$ and $S_{P'}$, is a word of length $\gamma n$. Let $c=(c_1, \ c_2, \ \cdots, \ c_n)$ be the code-word in $\cC$ given by the encoding of $w$. The evaluation of the distribution function $f$ at $c$ $$f(c) = (\{y_i, i \})_{1\leq i\leq n}, \text{ where } y_i = (c_i, c_{i+1}, \cdots, c_{i+\gamma})$$ gives to each participant $\cP_i$ a share $\{y_i, i \}$ of bits $y_i$ of length $\gamma$ and an index $i$. In the reconstruction phase, we define the following as an initial state:   
    $$ *** \cdots *** \cdots \underbrace{*** \cdots ***}_{ \gamma \ symbols}  \cdots *** \cdots **.$$
Above the brackets we have $\gamma$ symbols $*\in \{0, 1\}$. This part of the proof is handled by Procedure \ref{proc:SIP}.
    
\begin{algorithm}
\caption{SharingIsogenyPath}\label{proc:SIP}
\begin{algorithmic}[1]
\Require $I:E_0\longrightarrow E_1, \ P$, a code $ \cC=[\ell, \ k, d], n \geq 2$.
\Ensure $(\{\cP_i, \ i\})_{0\leq i\leq n-1}$  \Comment{$\cP_i$ is a chain of symbols from $\cC$}
\State $P' \gets I(P)$
\State $S_P \gets SESS(E_0, P)$
\State $S_{P'} \gets SESS(E_1, P')$
\State $w \gets S_P||S_{P'}$ 
\State $c \gets \cE(w)$  \Comment{$c = (c_0, \ c_1, \ \cdots, \ c_{\gamma n-1})$}
\State $\cP \gets ( \ )$ \Comment{A table containing the shares}
\State $i\gets 0$
\While{$i < n$}
    \State $\cP_i \gets ( \ )$
    \State $j \gets 0$
    \While{$j < \gamma$}
        \State $\cP_i+ \gets c_{i\gamma + j}  $            
    \EndWhile
    \State $\cP+ \gets \{i, \cP_i\}$ \Comment{Adding the secret share of participant $\cP_i$}
    \State $i\gets i+1$ 
\EndWhile
\State \textbf{return} $\cP$
\end{algorithmic}
\end{algorithm}

    On the other hand, when $t$ participants $\cP_i$ bring their shares $\{y_i, i \}_{1\leq i\leq t}$ we put each $y_i$ at the rith place using the index $i$. So it is important to record the index $i$ in order to have fewer possible errors. In this case, we can consider the errors as erasures because we know their exact positions. After the $t$ participants fill their shares in the initial state it remains $(n-t)$ block of unknown bits that can be viewed as erasures. So, in this case, we can use the algorithms of decoding of the code $\cC$ to correct the $\gamma(n-t)$ possible erasures. For the next step, we use the inverse of the algorithm \textit{SESS} to get $P$ and $P'$. Finally using this extra information, the algorithms of reduction of Feo et al. \cite[Section 5.4]{FFP24} and Theorem \ref{theo:IsoEfficiency}, the $t$ participants can efficiently extract the secret path $I$ as described in Procedure \ref{proc:RIP}.
\end{proof}

\begin{algorithm}
\caption{RecoverIsogenyPath}\label{proc:RIP}
\begin{algorithmic}[1]
\Require  $E_0, \ E_1, \ \{i, \ \cP_i\}_{i \in [0 \ \cdots \ n-1], \ \#i=t }, \ \cC, \ \cE^{-1}, \ \cD,  \ \gamma, \ \cA $
\Ensure $I:E_0 \rightarrow E_1$
\State $\mathfrak{e} \gets (e, e, \cdots, e)$ \Comment{$\mathfrak{e}$ contains 
$\gamma$ erasures $e$}
\State $y \gets ( \ ) $ 
\For{$0\leq j \leq n-1$}
   \If{$j\in \{i\}_{\#i = t}$} \Comment{Verify if $j$ is in the collection of $i$}
      \State $y+ \gets \cP_j$
   \Else 
       \State $y+ \gets \mathfrak e$ \Comment{The missing symbols are considered erasures}
   \EndIf
\EndFor
\State $c \gets \cE^{-1}$(y) \Comment{Correcting the erasures in $y$}
\State $w \gets \cD(c)$ \Comment{Decoding the code-word $c$}
\State $S_P\gets ( \ )$
\State $S_Q\gets ( \ )$
\For{$0\leq j\leq \gamma n$}
   \If{$j< \gamma n/2$}
      \State $S_P+\gets w_j$
    \Else
       \State $S_Q+\gets w_j$
   \EndIf
\EndFor
\State $P'\gets \psi^{-1}(S_P)$
\State $Q'\gets \psi^{-1}(S_Q)$
\State $I\gets \cA(E_0, \ E_1, \ P', \ Q')$ \Comment{Uses algorithms from Theorem \ref{theo:SI-OTP}} as a black box.
\State \textbf{return} $I$
\end{algorithmic}
\end{algorithm}

\section{Security of the Scheme}\label{sec:SS}
We recall that an $(n, t)$ threshold scheme is perfect if any subset of participants 
$(\cP_i)_{1\leq i\leq n}$ either gives all information about the shared secret $s$ or does not contain any information about $s$.

\begin{theorem}\label{theo:secscheme}
    Let $\mathcal{T}(n, t)$ be the isogeny-based threshold scheme given by a binary code of length $\ell$ and dimension $k$ where $\ell=\gamma n$. If
    \begin{align}\label{bound:SecScheme}
        \frac{k}{\gamma}\leq t\leq n-\frac{128}{\gamma}+1,
    \end{align}
    then $\cT(n, t)$ is a perfect threshold scheme in the context of NIST security level I.
\end{theorem}

\begin{myproof}
    By definition, $\cC$ can correct up to $\gamma(n-t)$ erasure and only $t$ or more participants can make the correction. If $t-1$ or fewer participants want to brute force the process, they have to perform $$2^{\gamma(n-t+1)}$$ attempts since the code can't correct $\gamma(n-t+1)$ erasures. When considering the system is secure for a length of bit equal to 128, as in NIST security level I, we have
    \begin{align*}
        \gamma(n-t+1) \geq 128 & \Rightarrow n-t \geq \frac{128}{\gamma}-1.\\
    \end{align*}
This implies $$t \leq n-\frac{128}{\gamma}+1.$$    
On the other hand, we know that, as a linear code $\cC$ can correct up to $d-1$ erasures with $d$ the minimum distance of $\cC$. According to the singleton bound, $d\leq \ell-k+1$. Since $\ell=\gamma n$, then we have
\begin{align*}
    \gamma(n-t)\leq \gamma n - k \Rightarrow t \geq \frac{k}{\gamma}.
\end{align*}
Finally, we get the bound (\ref{bound:SecScheme}).
\end{myproof}

\noindent From Theorem \ref{theo:secscheme} we can derive corollaries that align with NIST security levels III and V. Without loss of generality, in this paper, we focus on NIST I. \\
\indent Looking at the security bound given in (\ref{bound:SecScheme}) we can see that the number of possible participants $n$ and the threshold $t$ depend mainly on the parameter $\gamma$ in the distribution function $f$ and the dimension $k$ of the code $\cC$. So, for the following we denote by $\cT(n, t, \gamma)$ the threshold scheme in order to highlight the central role of $\gamma$ in the scheme.

\section{On Some Codes for the Scheme $\cT(n, t, \gamma)$}\label{sec:SCS}

To enlarge the bound (\ref{bound:SecScheme})  we should increase the parameter $\gamma$ as much as possible and work with sufficiently long codes with a high erasure correction capability. So, a good implementation of our scheme is to use distance-optimal codes\footnote{An $[\ell, k, d]$ code over $\FF_q$ is distance-optimal if there is no $[n, k, d']$ code over $\FF_q$ with $d' > d$} with small dimensions. Reed-Solomon codes and their subfield codes are good candidates.

\subsection{Binary Reed-Solomon Codes and Isogeny Sharing}\label{subsec:BRSIS}

A binary Reed-Solomon code $RS(2^r, d)$ is a cyclic linear code \cite[Chap. 7, Sec. 8]{MS78} over $\FF_{2^r}$ with minimum distance $d$ and generator $$g(x)=(x+\tau^{m+1})(x+\tau^{m+2})\cdots(x+\tau^{m+d-1})$$
for some integer $m$ and some primitive element $\tau$ of $\FF_{2^r}$. 
If $\cC$ is an $RS(2^r, d)$ code then the length $\ell = 2^r-1$ and the dimension $k=2^r-d$. Reed-Solomon codes are very interesting because they are MDS codes, efficient in error/erasure correction, and natural codes to use if the required length is less than the field size. Furthermore, as we can see in the following theorem they give flexibility on the choice of parameters.

\begin{theorem}[\cite{MS78}]\label{theo:rs_mds} 
Let $\FF_q$ be a finite field of cardinality $q$.
    \begin{enumerate}
        \item For any $k, \ 1 \leq k \leq q + 1$, there exists a $[\ell=q - 1, \ k, \ \ell-k+1]$ Reed-Solomon codes over $\FF_q$.
        \item There exist $[2^r + 2, \ 3, \ 2^r]$ and $[2^r+2, \ 2^r-1, \ 4]$ triply-extended RS that are MDS codes.
    \end{enumerate}
\end{theorem}

For RS code over $\FF_{2^r}$ usually the degree of extension $r>1$. Indeed the only binary MDS codes of length $\ell$ are the trivial ones, i.e., the codes with parameters $[\ell, 1, \ell], \ [\ell, \ell - 1, 2]$ and $[\ell , \ell, 1 ]$ (see \cite[Chapter 11]{MS78}).
Despite the flexibility given by RS codes, we can't use them directly in our scheme. That is because the output of the encoding of the secret torsion points is a chain of bits. To take advantage of RS codes in our scheme, we use subfield codes from RS codes.
\begin{lemma}
    Let  $[\ell,k,d]$ be a code $\cC$ over $\FF_{q^r}$, its subfield code $\cC^* = \cC|_{\FF_q}$ is defined as follows: $\cC^* = \cC\cap \FF_q^n.$ 
    The code $\cC^*=[\ell, k^*, d^*]$ is a linear code with $d\leq d^* \leq \ell$ and $\ell-k\leq \ell-k^*\leq r(\ell-k)$.
\end{lemma}
This lemma shows that, when derived from an $RS(2^r, d)$  code, the code $\cC^*$ will be of particular interest in the design of our scheme. Indeed it allows us to work with the same length $\ell$, to have a smaller dimension $k^*\leq k$ and a higher minimum distance $d^*\geq d$.\\
\indent As proposed by Heng and Ding \cite{HD22} to obtain very good subfield codes over small fields, e.g., $\FF_2$, we can choose good codes over an extension field $\FF_{2^r}$ with small dimensions. So, to share the isogeny path as described in Procedure \ref{proc:SIP} we can particularly use the triply-extended RS code given in points $2$ of Theorem \ref{theo:rs_mds}. These linear codes,  with parameters $[2^r + 2, \ 3, \ 2^r]$ over $\FF_{2^r}$ are called hyperoval codes. In \cite{HD19}, it is shown that the subfield codes of some hyperoval codes are distance-optimal, which makes them an appropriate choice in the design of our scheme. We recall that in this scheme the length $\ell$ of the code $\cC$ satisfies
$\ell = \gamma n $ which implies that the number of participants is $n=(2^r+2)/\gamma$ if we choose the code $[2^r + 2, \ 3, \ 2^r]$. We derive the following corollary from Theorem \ref{theo:secscheme}.
\begin{corollary}
    Let $\mathcal{T}(n, t, \gamma)$ be the isogeny-based threshold scheme given by a binary subfield code $[2^r + 2, \ k^*, \ d^*]$ of an hyperoval code $[2^r + 2, \ 3, \ 2^r]$ over $\FF_{2^r}$. If 
    $$\frac{2}{\gamma}\leq t\leq \frac{2^r-126}{\gamma}+1 \ \text{ or } \  \frac{3}{\gamma}\leq t\leq \frac{2^r-126}{\gamma}+1,$$ then we can share a secret isogeny path among $n$ participants with a threshold $t$.
\end{corollary}
Note that since the dimension $k^*$ of the subfield codes is less or equal to the dimension of the initial code $\cC$ then in this case $k^*$ is equal to $2$ or $3$.\\


\subsection{Burst Erasures and Secret Isogeny Recovery}\label{subsec:BESIR}

On many communication channels or storage systems, the errors/erasures are not random but tend to occur in clusters or bursts. A binary burst erasure refers to a type of erasure that occurs when multiple bits are affected simultaneously, rather than isolated single-bit erasures. A burst of length $\gamma$ is characterized by the fact that $\gamma$ consecutive bits are affected and should be corrected for further communication or processing. The threshold scheme $\cT$ designed in this paper is typically using burst erasures. So in this section, we leverage their properties and the binary representation of a $RS(2^r, d)$ code to give a corollary to the main theorems. \\
\indent Another way of using RS codes in our scheme is to consider their binary representation. For an $RS(2^r, d)$ code $\cC$ over $\FF_{2^r}$, we can map each codeword to its binary representation using the one-to-correspondence 
\[ \begin{array}{llll}
 \Phi: & \FF_{2^r} & \longrightarrow & \FF_2^r\\
 & c_i & \longmapsto & (c_{i1}, \ c_{i2}, \ \cdots, \ c_{ir})
\end{array} \]
between elements of $\FF_{2^r}$ and vectors given by scalars in $\FF_2^r$. For a codeword $c=(c_0, c_1, \cdots, c_{\ell-1})$ in $RS(2^r, d)$ we replace each $c_i$ by the corresponding vector. Adding an overall parity check on each $r-$tuple gives a code $\cC^+$ with parameters 
$$(r+1)(2^r-1), \ rk, \text{ and } D\geq 2d.$$
In addition to the merits enumerated in the previous section, Reed-Solomon codes have also two other advantages. Their binary mappings, $\cC^+$, can have a high minimum distance $D$ \cite[Chap. 10, Sec. 5]{MS78} and they are useful for correcting several bursts.

\begin{theorem}\cite[Chap. 10, Sec. 6]{MS78}\label{theo:BurstBound}
    Let $\cC^+$ be a binary code from an $RS(2^r, d)$ code. Then a binary burst of length $\gamma$ can affect at most $\epsilon$ adjacent symbols from $\FF_{2^r}$, where $\epsilon$ is given by
    $$(\epsilon-2)r + 2 \leq \gamma \leq (\epsilon - 1)r + 1.$$
    So if the minimum distance $D$ is much greater than $\epsilon$, many bursts can be corrected.
\end{theorem}

To make use of this result in our scheme we perform a back-and-forth between RS codes and their binary representations. 
\begin{enumerate}
    \item We consider the code $\cC$ used in Procedures \ref{proc:SIP} and \ref{proc:RIP} as a code $\cC^+$ from the binary representation of an $RS(2^r, d)$ Reed-Solomon code over $\FF_{2^r}$.
    \item We make a call to Procedure \ref{proc:SIP} to share $\gamma n$ bits among the $n$ participants.
    \item For the recovery of the isogeny
    \begin{itemize}
        \item for each cluster of $\gamma-$bits given by each of the $t$ participants we give the corresponding symbol in $\FF_{2^r}$ using the inverse of the map $\Phi$,
        \item for the missing $(n-t)$ clusters of $\gamma-$bits, we choose random bits in $\{0, 1\}$ then give the corresponding symbol in $\FF_{2^r}$. The obtained symbols are viewed as erasure in the $RS(2^r, d)$ code.
    \end{itemize}
    \item Using Theorem \ref{theo:BurstBound} we measure the number of symbols affected over $\FF_{2^r}$.
    \item We obtain a codeword $c$ in $RS(2^r, d)$ by correcting the erasures.
    \item We give the codeword $c^+$ that corresponds to the binary representation of $c$.
    \item We use the last part of Procedure \ref{proc:RIP} (i.e., decoding and extracting in lines 11 to 23) to recover the isogeny. \\
\end{enumerate} 

The importance of the above algorithm stems from the fact that we can directly work with RS codes over $\FF_{2^r}$ and easily get a binary representation of their symbols using a basis of $\FF_{2^r}$ over $\FF_2$. We summarize the result of this section in the following corollary.

\begin{corollary}
    Let $\cT(n, t, \gamma)$ be an isogeny-based threshold scheme as defined in Theorem \ref{theo:main}. Assume that the code $\cC$ is a binary representation $\cC^+$ of an $RS(2^r, d)$ code. If  $$r>\gamma-2 \text{ and } d\geq 2(n-t)+1$$ then we can share a secret isogeny among $n$ participants with a threshold $t$. 
\end{corollary}

\begin{myproof}
    From Theorem \ref{theo:BurstBound}, we know that the number of RS affected symbols from each binary burst erasure is $$\epsilon \leq \frac{\gamma-2}{r}+2.$$ So if $r>\gamma-2$ the burst length for the $RS(2^r, d)$ code is less or equal to $2$. The total number of erasures over $\FF_{2^r}$ is $\leq2(n-t)$.  By definition, Reed-Solomon codes are MDS and can correct up to $d-1$ erasures. So, if $$2(n-t)\leq d-1 \Leftrightarrow d\geq 2(n-t)+1,$$
    we can correct the erasures and use the previous algorithm to efficiently recover the isogeny path.
\end{myproof}

\bigskip

For the efficiency of the derived algorithms, we note the following. As a cyclic code, we can use efficient systematic encoders to encode words in the $RS(2^r, d)$ code and as a special case of BCH codes \cite[Chap. 9, Sec. 6]{MS78} we can use decoding algorithms pertaining to these types of codes for the decoding process.


\end{document}